\documentclass[a4paper,UKenglish]{lipics-v2021}

\hideLIPIcs  
\nolinenumbers

\title{Automata on $S$-adic words} 

\author {Val\'erie Berth\'e} 
{Universit\'e Paris Cité, IRIF, CNRS, Paris, France}
{berthe@irif.fr}
{https://orcid.org/0000-0001-5561-7882}
{}
\author {Toghrul Karimov} 
{Max Planck Institute for Software Systems, Saarland Informatics Campus, Saarbr\"ucken, Germany}
{toghs@mpi-sws.org}
{https://orcid.org/0000-0002-9405-2332}
{}
\author {Mihir Vahanwala} 
{Max Planck Institute for Software Systems, Saarland Informatics Campus, Saarbr\"ucken, Germany}
{mvahanwa@mpi-sws.org}
{https://orcid.org/0009-0008-5709-899X}
{}

\authorrunning{V. Berth\'e et al.} 

\Copyright{Val\'erie Berth\'e, Toghrul Karimov, Mihir Vahanwala} 



\category{} 

\relatedversion{} 

\EventEditors{John Q. Open and Joan R. Access}
\EventNoEds{2}
\EventLongTitle{42nd Conference on Very Important Topics (CVIT 2016)}
\EventShortTitle{CVIT 2016}
\EventAcronym{CVIT}
\EventYear{2016}
\EventDate{December 24--27, 2016}
\EventLocation{Little Whinging, United Kingdom}
\EventLogo{}
\SeriesVolume{42}
\ArticleNo{23}

\newcommand{\nat}{\mathbb{N}}
\newcommand{\intg}{\mathbb{Z}}

\newcommand{\rat}{\mathbb{Q}}

\newcommand{\Acal}{\mathcal{A}}
\newcommand{\Bcal}{\mathcal{B}}
\newcommand{\Ccal}{\mathcal{C}}

\newcommand{\Lcal}{\mathcal{L}}

\newcommand{\Wcal}{\mathcal{W}}
\newcommand{\Xcal}{\mathcal{X}}

\usepackage{mathtools}
\begin{document}

\maketitle

\begin{abstract}
    A fundamental question in logic and verification is the following: for which unary predicates $P_1, \ldots, P_k$ is the monadic second-order theory of $\langle \mathbb{N}; <, P_1, \ldots, P_k  \rangle$ decidable? Equivalently, for which infinite words $\alpha$ can we decide whether a given B\"uchi automaton $\Acal$ accepts $\alpha$? 
    Carton and Thomas showed decidability in case $\alpha$ is a fixed point of a letter-to-word substitution $\sigma$, i.e., $\sigma(\alpha) = \alpha$. 
    However, abundantly more words, e.g., Sturmian words, are characterised by a broader notion of self-similarity that uses a set $S$ of substitutions. 
    A word $\alpha$ is said to be directed by a sequence $s = (\sigma_n)_{n \in \mathbb{N}}$ over $S$ if there is a sequence of words $(\alpha_n)_{n \in \mathbb{N}}$ such that $\alpha_0 = \alpha$ and $\alpha_n = \sigma_n(\alpha_{n+1})$ for all $n$; such $\alpha$ is called $S$-adic. 
    We study the automaton acceptance problem for such words and prove, among others, the following.
    Given finite $S$ and an automaton~$\Acal$, we can compute an automaton $\Bcal$ that accepts $s \in S^\omega$ if and only if $s$ directs a word $\alpha$ accepted by $\Acal$. Thus we can algorithmically answer questions of the form  ``Which  $S$-adic words are accepted by a given automaton $\Acal$?''
\end{abstract}

\section{Introduction}
In 1962, Büchi proved that the monadic second-order (MSO) theory of the structure $\langle\mathbb{N}; <\rangle$ is decidable \cite{buchi-seminal}, and, in doing so, laid the foundations of the theory of automata over infinite words. 
Subsequently, Elgot and Rabin \cite{elgot-rabin} adopted automata-theoretic techniques to show how to decide the MSO theory of $\langle\mathbb{N}; <, P\rangle$ for various interesting unary predicates~$P$ including $\{n! \colon n\in\nat\}$ and $\{2^n \colon n\in\nat\}$.
By then, it was already known that unary predicates lay at the frontiers of decidability: expanding $\langle \nat; < \rangle$ with most natural functions (e.g., addition or doubling) or non-unary predicates yields undecidable MSO theories \cite{robinson1958restricted,trahtenbrot1962finite,thomas1975note}. 
The question thus arose: for which unary predicates $P_1,\ldots,P_k$ is the MSO theory of $\langle\mathbb{N}; <, P_1, \ldots, P_k\rangle$ decidable? 
Equivalently, 
for which infinite words $\alpha$ is the \emph{automaton acceptance problem\footnote{Here the problem is parametrised by $\alpha$; the only input is the automaton $\Acal$.}}, which asks whether a given automaton $\Acal$ accepts $\alpha$, decidable?

The automaton acceptance problem under various assumptions on $\alpha$ has been studied, among others, by Semënov, Carton and Thomas, and Rabinovich \cite{semenov1984logical,carton-thomas,rabinovich}.
Semënov \cite{semenov1984logical} showed decidability for $\alpha$ that are \emph{effectively almost-periodic}.
These include, for example, the Thue-Morse word and \emph{toric words}, which are obtained from certain compact dynamical systems~\cite{berthe2025monadic}.
Carton and Thomas \cite{carton-thomas}, on the other hand, used algebraic methods to show decidability for \emph{morphic} $\alpha$, which include $\alpha$ that can be constructed by infinitely iterating a letter-to-word morphism $\sigma$ on a starting letter $a$.
Their result implies, in one fell swoop, the decidability of the MSO theory of $\langle \nat; <, \{p(n)a^n \colon n \in \nat \}\rangle$ where $a \ge 1$ and $p$ is a polynomial with integer coefficients satisfying $p(\nat) \subseteq~\nat$.

By now, the study of the automaton acceptance problem for various special classes of~$\alpha$ has led to rich interactions between word combinatorics, algebra (particularly monoid and group theory), number theory, and formal verification.
For example, effectively almost-periodic words of Sem\"enov have been identified as a powerful tool for analysing the behaviour of linear $\mathsf{while}$ loops in program verification \cite{karimov2022s}.
More recently, \cite{lics24-mso} showed the decidability of the MSO theory of $\langle \mathbb{N}; <, a^\mathbb{N}, b^\mathbb{N}\rangle$ using the fact that the order in which powers of coprime $a$ and $b$ occur is captured by a certain \emph{Sturmian word} (Sec.~\ref{subsection:Sturmian}). 
Sturmian words are an extremely well-studied and fundamental class of ``special words'' that appear naturally in a range of fields including number theory, computer graphics, and astronomy \cite[Chap.~9.6]{allouche_shallit_2003}, \cite[Chap.~6]{fogg2002substitutions}.
We refer the reader to \cite{berthe2025monadic} for a more detailed survey of the role of Sturmian words, and word combinatorics in general, in logic and verification.

In this paper, we are motivated by questions of the form, ``Given an $\omega$-regular language~$L$, decide whether it contains a Sturmian word.'' 
(Of course, we can also ask about any other combinatorial class.)
Our approach adopts the \emph{$S$-adic perspective}, a powerful tool for elucidating combinatorial properties of infinite words.
Akin to the continued fraction expansion of a real number, or even a Fourier decomposition of a signal, we write an infinite word $\alpha$ as an infinite composition of (possibly different) \emph{substitutions}.
A substitution $\sigma$ over an alphabet $\Sigma$ gives rules to replace each letter $a \in \Sigma$ with a non-empty word $\sigma(a) \in \Sigma^+$.
For example, the Fibonacci substitution $\sigma_\mathsf{fib}$ over $\{0, 1\}$ replaces $0$ with $01$ and $1$ with $0$; the Fibonacci word $\alpha_{\mathsf{fib}} = 01001010\cdots$ is obtained as the limit of iterating $\sigma_\mathsf{fib}$ infinitely on the letter~$0$ (or, alternatively, the letter 1). 
Hence we have the infinite decomposition $\alpha_{\mathsf{fib}} = \sigma_{\mathsf{fib}} \circ \cdots \circ \sigma_{\mathsf{fib}} \cdots$.
In general, we have a set $S$ of substitutions, and say that a sequence $s$ over $S$ \emph{directs} $\alpha \in \Sigma^\omega$ if there exists a sequence of words $\left(\alpha^{(n)}\right)_{n \in \mathbb{N}}$ such that $\alpha^{(0)} = \alpha$, and $\alpha^{(n)} = \sigma_n\left(\alpha^{(n+1)}\right)$ for all~$n$.
This gives us the $S$-adic decomposition $\alpha = \sigma_0 \circ \sigma_1 \circ \cdots$.\footnote{Of course, we are not interested in trivial decompositions that, for example, just permute the letters back and forth, as these do not tell us anything new about $\alpha$.}
We refer to any $s \in S^\omega$ as a \emph{directive sequence}.

Entire classes of words such as Sturmian and Arnoux-Rauzy words (Sec.\ \ref{subsection:Sturmian}) can be defined in terms of directive sequences over a specific $S$;
such a class $\Wcal$ is called \emph{$S$-adic}. 
Our central question is the following: 
given an $\omega$-regular language $L \subseteq \Sigma^\omega$ and an $S$-adic class of words $\Wcal \subseteq \Sigma^\omega$, what is the set of all $\alpha \in \Wcal$ 
that are contained in~$L$? Does there exist at least one such $\alpha$? Equivalently, given an MSO formula $\varphi$, can we decide whether there exists $\alpha \in \Wcal$ that induces a structure in which $\varphi$ holds?
We remark that this question is similar in spirit to the problem solved in \cite{hieronymi2024decidability}: given a \emph{first-order} formula~$\varphi$, it is decidable whether there exists a Sturmian word $\alpha$ such that $\varphi$ holds in the induced first-order structure $\langle\mathbb{N}; <, +, P_\alpha \coloneqq \{n \colon \alpha(n)=1\}\rangle$.
The key idea there is that the class of structures $\langle \nat; <, +, P_\alpha\rangle$, where $\alpha$ ranges over all Sturmian words, is \emph{automatic} \cite{shallit2022logical}.


Our main contribution is that $\omega$-regular specifications on words in the $\Sigma$-space translate to $\omega$-regular constraints on the directive sequences in the original  $S$-space of substitutions. 
Hence the common MSO theory of an $S$-adic class of words is decidable.
This generalises the result of \cite{competing-result} that, given finite $S$ and an automaton $\Acal$ whose language is closed (i.e., a B\"uchi automaton whose states are all accepting), it is decidable whether $\Acal$ accepts a word directed by $S$. 
Our algorithms thoroughly answer questions of the kind ``Which $S$-adic words are accepted by a given automaton $\Acal$?''
For example, given two words $u,v \in \{0,\ldots,d-1\}^*$ and  $a,b$ with $b\ne 0$, we can compute an effective representation (as an $\omega$-regular language of directive sequences) of all Arnoux-Rauzy words in which between any two consecutive occurrences of $u$, the number of occurrences of $v$ is congruent to $a$ modulo $b$.

\subsection*{Outline and contributions of the paper}
In Sec.~\ref{sec:prelims} we establish the necessary mathematical background.
In Sec.\ \ref{sec::sadic-general} we formally define what it means for a directive sequence to generate and to direct a word.
Briefly, a sequence of substitutions $(\sigma_n)_{n\in\nat}$ generates $\alpha$ if there exists a sequence of letters $(a_n)_{n\in\nat}$ such that $\alpha = \lim_{n \rightarrow \infty} \sigma_0 \cdots \sigma_n(a_n)$; generating a word is a strong form of directing it.
We then recall well-known properties of directive sequences, the most important of them being \emph{weak primitivity}, and then describe Sturmian and Arnoux-Rauzy words, which are the most-well known examples of $S$-adic words (see e.g.\ \cite{fogg2002substitutions,Berthe-Delecroix}).

In Sec.~\ref{sec::structure-theorems} we study the structure of words directed or generated by directive sequences. 
Our key new insight is the augmentation of a directive sequence $s$ over\footnote{Our set $S$ of substitutions could possibly be infinite.} $S$ into a \emph{congenial} expansion $\hat{s}$ over $S \times \Sigma$ (Def.~\ref{congenial-definition}), which generates a word incrementally and predictably (Lem.~\ref{lem::congenial-limit}). 
We prove the following pivotal results.
\begin{itemize}
    \item If $s$ generates $\alpha$, then it also directs $\alpha$ (Lem.\ \ref{generated-implies-directed}).
    \item A word $\alpha$ directed by $s$ is a concatenation of words generated congenially by $s$ (Lem.\ \ref{lem::congenial-concatenate}). 
    \item For $s$ weakly primitive, $s$ congenially generates $\alpha$ if and only if it directs $\alpha$ (Lem.\ \ref{lem:eg-sadic=sdirected}).
\end{itemize}

In Sec.\ \ref{sec:S-mod-A} we introduce a suitable equivalence relation of finite index on substitutions modulo a semi-group associated with a given $\omega$-regular language~$L$, inspired by the syntactic monoid of~$L$, that behaves well with respect to infinite composition of substitutions; we denote the set of classes by~$\Xi_L$. 
Given a sequence $s$ over  the set of substitutions $S$ or a sequence $\hat{s}$ over $S \times \Sigma$, we naturally define its trace to be, respectively, a word over $\Xi_L$ or over $\Xi_L \times \Sigma$. 
In Sec.\  \ref{section-sigma-omega} and \ref{sec::aut-acceptance} we use our two key ingredients, the notion of congeniality together with the monoid  $\Xi_L$ of equivalence classes, to prove our main results.
\begin{description}
    \item[Morphic Words] Let $\Acal$ be an automaton over $\Sigma$, and $\sigma, \pi$ be substitutions. 
    Using only the respective equivalence classes $\xi, \zeta\in \Xi_L$ of the latter, we can compute a regular language $L \subseteq \Sigma^+$ such that the word $\pi( \sigma^\omega(u))$ is well-defined and accepted by $\Acal$ if and only if $u \in L$ (Thm.\ \ref{morphic-regular-lang}).
    We can thus characterise all such $\pi,\sigma,u$, which generalises the result of \cite{carton-thomas} that it is decidable whether a given morphic word is accepted by $\Acal$.
    \item[Generated Words] Given an automaton $\Acal$ over $\Sigma$, we can construct an automaton $\Bcal$ over $\Xi_L \times \Sigma$ such that $\Bcal$ accepts the trace of $\hat{s}$ if and only if $\hat{s}$ is congenial and generates a word accepted by~$\Acal$ (Thm.~\ref{cor::congenial-automaton-version}).
    \item[Directed Words] Given $\Acal$ as above, we can construct an automaton $\Bcal$ over $\Xi_L$ such that $\Bcal$ accepts the trace of $s$ if and only if $s$ directs a word accepted by~$\Acal$ (Thm.\ \ref{desubstitutible-dream}).
    This further generalises the result of \cite{carton-thomas} from infinite compositions of a single substitution $\sigma$ to arbitrary infinite compositions over a set of substitutions $S$.
\end{description}

In Sec. \ref{sec:partial-quotients}, we refine our main results for Sturmian and Arnoux-Rauzy words, which have an \emph{a priori} known \emph{factor complexity} (Sec.\ \ref{sec:uniformly-recurrent}). 
We show that for such classes, acceptance by~$\Acal$ is completely determined by the first $N(\mathcal{A})$ \emph{partial quotients} of the directive sequence (Thm.~\ref{thm::partial-quotients}). 
In the case of Sturmian words, this has a nice geometric interpretation: an automaton can only resolve the slope and intercept associated with a Sturmian word up to a ``pre-determined'' finite precision.

\section{Preliminaries}
\label{sec:prelims}
An alphabet $\Sigma$ is a finite and non-empty set of symbols.
We write $\varepsilon$ for the empty word.
For a word $\alpha$, $\alpha(j)$ denotes the letter at the $j$th position of $\alpha$, $\alpha[i, j)$ denotes the finite word $\alpha(i) \cdots \alpha(j-1)$, and $\alpha[j, \infty)$ denotes the infinite word $\alpha(j)\alpha(j+1)\cdots$. A finite word $u$ is a factor of a word $v$ if there exist indices $i, j$ such that $v[i, j) = u$.
When we say that an object is effectively computable, we mean that a representation in a scheme (that will be clear from the context) is effectively computable. 

A substitution $\sigma$ is a non-erasing morphism from $\Sigma^*$ to $\Sigma^*$, i.e., $\sigma(v) = \varepsilon$ if and only if $v = \varepsilon$.
We denote the set of all such substitutions by $S(\Sigma)$.
For substitutions $\mu, \sigma$, we write $\mu \sigma$ for $\mu \circ \sigma$.
A substitution is positive if every $b \in \Sigma$ appears in $\sigma(a)$ for all $a \in \Sigma$, and left-proper if the images of all letters by $\sigma$ begin with the same letter.
For more on the subject of substitutions, see e.g.\ \cite{fogg2002substitutions}.

\subsection{Topology of finite and infinite words}
\label{prelims-topology}
We equip $\Sigma^\infty \coloneqq \Sigma^+ \cup \Sigma^\omega$ with the product topology, and define the distance between words $u, v$ to be $2^{-n}$, where $n$ is the first position in which they differ. E.g., distinct $a, b \in \Sigma$ are a distance of $2^0 = 1$ apart.
A notion of convergence of sequences of words follows naturally.

\begin{definition}
Let $(u_n)_{n\in\nat}$ be a sequence of finite non-empty words. We define $\alpha = \lim_{n\rightarrow \infty} u_n \in  \Sigma^\infty \cup \{\bot\}$ as follows.
\begin{itemize}
    \item If there exists $v \in \Sigma^*$ and $N$ such that for all $n \ge N$, $u_n = v$, then $\alpha = v$. 
    \item If there exists $\beta \in \Sigma^\omega$ such that for all $j$, $u_n[0, j) = \beta[0, j)$ for all sufficiently large $n$, then $\alpha = \beta$.
    \item Otherwise, $\alpha = \bot$, which denotes lack of convergence in $\Sigma^\infty$.
\end{itemize}
\end{definition}
Under this topology, $\Sigma^\infty$ is compact.
The cylinder sets defined by fixing finitely many letters are both closed and open. 

\subsection{Automata and semigroups for infinite words}
\label{sec::aut-semigroups}
We consider infinite words and $\omega$-regular languages from the algebraic and combinatorial perspectives. 
A language $L \subseteq \Sigma^\omega$ is $\omega$-regular if and only if it can be recognised by a (nondeterministic) Büchi automaton $\Acal = (\Sigma, Q, I, \Delta, F)$, where $\Sigma$ is the alphabet, $Q$ is the finite set of states, $I \subseteq Q$ is the set of initial states, $\Delta \subseteq Q \times \Sigma \times Q$ is the transition relation, and $F$ is the set of accepting states. A run $r \in Q^\omega$ of the automaton on an input word $\alpha$ satisfies $r(0) \in I$, and for all $n$, $(r(n), \alpha(n), r(n+1)) \in \Delta$. A word $\alpha$ is accepted if it has a run $r$ such that $r(n) \in F$ for infinitely many $n$. 

Some of our technical tools, e.g.\ Semënov's theorem, require the automaton to be \emph{deterministic}, i.e., there must be a single initial state, and the transition relation must induce a function $\delta: Q \times \Sigma \rightarrow Q$. For this reason, we also use deterministic parity automata, which are further equipped with a function $\mathsf{index}: Q \rightarrow \mathbb{N}$. In the case of deterministic parity automata, a word $\alpha$ has a single run $r$, and is accepted if $\limsup_{n \to \infty} \mathsf{index}(r(n))$ is even. Deterministic parity automata recognise precisely the class of $\omega$-regular languages. Thus, in most contexts, ``automaton'' can interchangeably be taken to mean nondeterministic Büchi automaton, or deterministic parity automaton.

We now recall $\omega$-semigroups (see \cite[Sec.~7]{perrin1995semigroups}) as an equivalent way to recognise $\omega$-regular languages. Formally, \emph{an $\omega$-semigroup} $M = (M_f, M_\omega)$ is a two-sorted algebra equipped with the following operations:
\begin{enumerate}
    \item A binary operation defined on $M_f$ and denoted multiplicatively,
    \item A mapping $M_f \times M_\omega \rightarrow M_\omega$, called the \emph{mixed product}, also denoted multiplicatively,
    \item An infinite product $\pi$ that maps infinite sequences over $M_f$ to an element of $M_\omega$.
\end{enumerate}
These operations must satisfy the following associativity properties:
\begin{enumerate}
\item $M_f$, equipped with the binary operation, is a semigroup.
\item For every $m_1, m_2 \in M_f$ and $m_3 \in M_\omega$, we have that $(m_1 m_2)m_3 = m_1(m_2 m_3)$.
\item For every sequence $(m_n)_{n \in \nat}$ over $M_f$, and every strictly increasing sequence $(k_n)_{n \in \nat}$ of indices, we have that
$\pi(m_0, m_1, m_2, \ldots) = \pi(m_0\cdots m_{k_0}, m_{k_0+1}\cdots m_{k_1}, \ldots)$.
\item For every $m \in M_f$ and every sequence $(m_n)_{n \in \nat}$ over $M_f$, we have that $m\cdot \pi(m_0, m_1, \ldots) = \pi(m, m_0, m_1, \cdots)$.
\end{enumerate}
The last two conditions show that the infinite product $\pi(m_0, m_1, \ldots)$ can also be denoted multiplicatively as $m_0m_1\cdots$. Indeed, an $\omega$-semigroup can be intuited as a semigroup where infinite products are defined. An immediate example is $\Sigma^\infty$ with word concatenation, where $M_f = \Sigma^+$, and $M_\omega = \Sigma^\omega$.

Given $\omega$-semigroups $M_1, M_2$ a morphism $h$ of $\omega$-semigroups is a pair $h_f, h_\omega$ such that $h_f$ is a semigroup morphism from $M_{1,f}$ into $M_{2,f}$, and $h_\omega$ is a map from $M_{1,\omega}$ to $M_{2,\omega}$ preserving the mixed product and infinite product, i.e., for every $m_1 \in M_{1, f}$ and $m_2 \in M_{1, \omega}$, $h_f(m_1)h_\omega(m_2) = h_\omega(m_1 m_2)$, and for every sequence $(m_n)_{n\in \nat}$ over $M_{1, f}$, we have $h_f(m_0)h_f(m_1) \cdots = h_\omega(m_0 m_1 \cdots)$. We shall thus omit subscripts and denote the application of a morphism by simply $h$. As an immediate example, a non-erasing substitution $\sigma$ defines a morphism from $\Sigma^\infty$ to $\Sigma^\infty$.

A set $L \subseteq \Sigma^\infty$ is recognised by a morphism $h$ from $\Sigma^\infty$ into an $\omega$-semigroup $M$ if there exists a subset $H \subseteq M_f \cup M_\omega$ such that $L = h^{-1}(H)$. A language $L \subseteq \Sigma^\omega$ is $\omega$-regular if and only if it is recognisable by a finite $\omega$-semigroup \cite[Thm.~7.6]{perrin1995semigroups}. See \cite[Sections~8,9]{perrin1995semigroups} respectively for effective translations from Büchi automata to $\omega$-semigroups, and vice versa. Furthermore, analogous to the finite-word case, there is a notion of the (finite) syntactic $\omega$-semigroup of an $\omega$-regular language \cite[Sec.~11]{perrin1995semigroups}.

A few remarks addressing presentation concerns are in order. A result of Wilke \cite{wilke1991eilenberg} showed that for finite $\omega$-semigroups $M$, the infinite product is fully determined by the function $M_f \rightarrow M_\omega$ that maps $m$ to $m^\omega$ \cite[Thm.~7.1]{perrin1995semigroups}. This combinatorial result is proven through Ramsey's theorem, and gives a method to construct $\omega$-semigroups through Wilke algebras. We refer the reader to \cite[Sec.~7]{perrin1995semigroups} for details of how the semigroup structure of $M_f$ determines an extension into an $\omega$-semigroup through \emph{linked pairs}.

We use the equivalence of finite $\omega$-semigroups and automata through the following lemma. 
\begin{lemma}
\label{lem::infinite-product-in-monoids-via-automata}
    Let $M = (M_f, M_\omega)$ be a finite $\omega$-semigroup, and $x \in M_f \cup M_\omega$. Extend $M_f$ into a monoid $M_f'$ by adjoining a distinguished neutral element $1_M$: for all $m \in M_f'$, $1_M \cdot m = m \cdot 1_M = m$. We can construct an automaton $\Acal_x$ over the alphabet $M_f'$ that accepts $m_0 m_1 \cdots \in (M_f')^\omega$ if and only if the infinite product $m_0 m_1 \cdots$ (defined as the possibly finite product of the subsequence $m_{i_0}m_{i_1}\cdots$ obtained by discarding all $1_M$ terms) equals $x$.
\end{lemma}

\subsection{Uniformly recurrent words}
\label{sec:uniformly-recurrent}
For a word $\alpha \in \Sigma^\omega$, the set $\Lcal(\alpha) = \{u : u \text{ is a factor of } \alpha\}$ is called the \emph{(factor) language} of~$\alpha$. 
The factor complexity function $p_\alpha$ computes the number of factors of $\alpha$ of a given length~$n$. 
For example, if $\alpha$ is Sturmian, then $p_\alpha(n) = n+1$ for all $n$ (see e.g.\ \cite{fogg2002substitutions}).

For an infinite word $\alpha$ and $l \ge 0$, denote by $R_\alpha(l)$ the smallest $r \in \nat \cup \{\infty\}$ such that every factor of $\alpha$ of length $l$ is a factor of every factor of $\alpha$ of length $r$.
We call $R_\alpha$ the \emph{recurrence function} of $\alpha$.
A word $\alpha \in \Sigma^\omega$ is said to be \emph{uniformly recurrent} if $R_\alpha(l) \in \nat$ for every $l \in \nat$.
That is, every $u \in \Sigma^l$ either does not occur in $\alpha$, or occurs infinitely often with bounded gaps.
It is clear from the definition that for uniformly recurrent $\alpha$, the value of $R_\alpha(l)$ only depends on $\Lcal(\alpha)$.
We record the following, which follows by brute enumeration.
\begin{lemma}
    \label{lemma:language-to-return-time}
    Let $\alpha \in \Sigma^\omega$ be uniformly recurrent. 
    Suppose we have access to an oracle that, given $u \in \Sigma^*$, checks whether $u \in \Lcal(\alpha)$.
    Then we can effectively compute $R_\alpha(l)$.
\end{lemma}
Sem\"enov \cite{semenov1984logical} gave an algorithm\footnote{Sem\"enov's result applies to the more general family of \emph{effectively almost-periodic words} \cite{semenov-modern}.} for determining whether a given deterministic automaton $\Acal$ accepts a given uniformly recurrent word~$\alpha$, which is represented by (A) an oracle computing $\alpha(n)$ on input $n$ and (B) an oracle computing an upper bound $\overline{R}_\alpha(l)$ on $R_\alpha(l)$ given $l$.
\begin{theorem}[Semënov]
    Let $\mathcal{A}$ be an automaton and $\alpha$ be a uniformly recurrent word represented by the oracles (A-B). 
    We can effectively compute $M \in \nat$ such that a state $q$ of $\Acal$ appears infinitely often in $\Acal(\alpha)$ if and only if it appears in $\Acal(\alpha)[M,2M)$.
\end{theorem}
In other words, to check whether $\Acal$ accepts $\alpha$, we simply need to run $\Acal$ on $\alpha$ for $2M$ steps, and observe the states that are visited.
See \cite[Chap.~3.1]{karimov-thesis} for effective bounds on $M$, from which we can deduce the following.
\begin{proposition}
\label{prop::semenov-and-prefix}
    Let $\mathcal{A}$ be an automaton and $\alpha$ be a uniformly recurrent word represented by the oracles (A-B).  
    We can compute $l, M$ such that for any uniformly recurrent $\beta$ with $\beta[0, M) = \alpha[0, M)$ and $R_\beta(n) = R_\alpha(n)$ for all $n \le l$, we have that $\Acal$ accepts $\alpha$ if and only if it accepts $\beta$. 
\end{proposition}

\subsection{Ostrowski numeration systems}
We recall properties of a numeration system which proves to be particularly  convenient  to   describe  Sturmian words in $S$-adic terms.
Let $\eta \in (0,1) \setminus \rat$.
The \emph{continued fraction expansion} of $\eta$ is the unique sequence $(a_n)_{n \ge 1}$ of positive integers such that
\[
\eta = \cfrac{1}{a_1 + 
    \cfrac{1}{a_2 + \cdots}
}
\]
We write $\eta = [0; a_1, a_2, \ldots]$.
The \emph{convergents} $(p_n/q_n)_{n \in \nat}$ of $\eta$ are obtained by truncating the expansion at the $n$-th level. 
The numerators and denominators satisfy the recurrences $p_0 = 0$, $q_0 = 1$, $p_1 = 1$, $q_1 = a_1$, and $(p_{n+2},q_{n+2}) = a_{n+2} \cdot (p_{n+1},q_{n+1}) + (p_n,q_n)$ for all $n \ge 0$.
The convergents are the \emph{locally best approximants} of $\eta$: for every $n \in \nat$, $p \in \intg$, and $0 < q < q_n$,
\[
|q_n \eta - p_n| < \min_{p \in \nat} |q \eta - p|
\]
which implies that $\big| \eta - \frac{p_n}{q_n} \big| < \min_{p \in \nat} \big| \eta - \frac p q \big|$.
The \emph{Ostrowski numeration system in base $\eta$} is based on 
the sequence $\theta_n = q_n \eta - p_n$.
For any $\chi \in [-\eta, 1-\eta]$, there exists a sequence $(b_n)_{n \ge 1}$ over $\nat$ such that (i) $0 \le b_1 < a_1$, (ii) $0 \le b_n \le a_n$ for all $n \ge 2$, (iii) for all $n$, $b_n = 0$ if $b_{n+1}=a_{n+1}$, and
\[
\chi = \sum_{n=1}^\infty b_n \theta_{n-1}.
\]
We refer to $(b_n)_{n\ge 1}$ as an \emph{Ostrowski expansion of $\chi$ in base~$\eta$}.
Conversely, every $(b_n)_{n\in\nat}$ satisfying (i-iii) is an Ostrowski expansion of some $\chi$ in base $\eta$, i.e., the infinite sum converges to a value in $[-\eta, 1-\eta]$.\footnote{To check this, observe that $\theta_n$ alternates between positive and negative; $\theta_0 = \eta$, $\theta_1 = a_1 \eta - 1$; $\theta_{n+2} = a_{n+2}\theta_{n+1} + \theta_n$.}
If $\chi \notin \intg + \eta \intg$ or if $\chi \in \intg_{\geq 1}  + \eta \intg$, then $\chi$ has a unique Ostrowski expansion in base $\eta$.
Otherwise, $\chi$ can have two expansions in base $\eta$. For more on the subject, see e.g.\ \cite{BugeaudLaurent}.

\section{\texorpdfstring{$S$-adic words}{S-adic words}}
\label{sec::sadic-general}

We now establish basic definitions and facts about $S$-adicity; see e.g.\ \cite{Berthe-Delecroix} for more  on the subject.
Let $\Sigma$ be an alphabet and $S \subseteq S(\Sigma)$ be a possibly infinite set of non-erasing substitutions.
We refer to $\alpha \in \Sigma^* \cup \Sigma^\omega$ as \emph{$S$-directed} if there exists a sequence $(\sigma_n)_{n\in\nat}$ over $S$ and a sequence of infinite words $(\alpha^{(n)})_{n\in\nat}$ such that $\alpha^{(0)} = \alpha$ and $\sigma_n(\alpha^{(n+1)}) = \alpha^{(n)}$ for all $n \in \nat$; note that each word in the sequence is the image of the subsequent one.
We say that $(\sigma_n)_{n\in\nat}$ \emph{directs} $\alpha$.
A word $\alpha \in \Sigma^* \cup \Sigma^\omega$ is called \emph{$S$-generated} if there exists a sequence $(\sigma_n, a_n)_{n\in\nat}$ over $S(\Sigma) \times \Sigma$ such that 
\begin{equation}
    \label{eq::s-adic1}
    \alpha = \lim_{n \to \infty} \sigma_0 \cdots \sigma_n(a_n).
\end{equation}
We refer to $(\sigma_n, a_n)_{n\in\nat}$ as an \emph{$S$-adic expansion} of $\alpha$.
For finite or infinite $\alpha$, whenever~\eqref{eq::s-adic1} holds we say that $(\sigma_n, a_n)_{n\in\nat}$ \emph{generates}~$\alpha$.
In both the $S$-directed and the $S$-generated settings (to which we collectively refer as \emph{$S$-adic}), $(\sigma_n)_{n\in\nat}$ is called a \emph{directive sequence}.

\begin{lemma}
\label{generated-implies-directed}
    If $\alpha \in \Sigma^* \cup \Sigma^\omega$ is generated by $s$ over $S(\Sigma)$, then it is also directed by $s$. 
    Furthermore, every $s$ over $S(\Sigma)$ directs at least one non-empty word $\beta \in \Sigma^+ \cup \Sigma^\omega$.
\end{lemma}
\begin{proof}
We will first prove a slightly more general version of the first statement.
Let $(a_n)_{n \in \nat}$ be a sequence of letters and $U_m = \{\sigma_m \cdots \sigma_n(a_n) : n \ge m\}$ for all $m \in \nat$.
Suppose $\alpha$ is an accumulation point of $U_0$, which subsumes the case of $\alpha$ being generated by $(\sigma_n,a_n)_{n\in\nat}$.
We will inductively prove the existence of $(\alpha^{(m)})_{m \in \mathbb{N}}$ such that $\alpha^{(0)} = \alpha$, $\alpha^{(m)} = \sigma_m\left(\alpha^{(m+1)}\right)$, and $\alpha^{(m)}$ is an accumulation point of $U_m$ for all $m$.
The base case is immediate.

For the inductive step, suppose we have constructed $\alpha^{(0)},\ldots,\alpha^{(m)}$ with the properties above. 
Write $u_{m, n} = \sigma_m \cdots \sigma_n(a_n)$, and observe that by the induction hypothesis, $\alpha^{(m)}$ is the limit of some sequence $(\sigma_m(u_{m+1, n_i}))_{i \in \mathbb{N}}$. 
By compactness, the sequence $(u_{m+1, n_i})_{i \in \nat}$ itself has an infinite subsequence $(u_{m+1, k_{j}})_{j \in \mathbb{N}}$ that converges. 
We choose $\alpha^{(m+1)}$ to be the limit, which is an accumulation point of $U_{m+1}$. 
By the continuity of $\sigma_m \colon \Sigma^* \cup \Sigma^\omega \to \Sigma^* \cup \Sigma^\omega$,
\[
\lim_{j \rightarrow\infty}\sigma_m(u_{m+1, k_{j}}) 
= \sigma_m\
\bigg(
\lim_{j \rightarrow\infty}u_{m+1, k_{j}}
\bigg)
\]
which implies that $\alpha^{(m)} = \sigma_m(\alpha^{(m+1)})$.

To prove the second claim, choose an arbitrary sequence $(a_n)_{n\in\nat}$ of letters and let $\beta$ be an accumulation point of $\{\sigma_0\cdots\sigma_n(a_n) \colon n \in\nat\}$.
Apply the preceding argument.
\end{proof}
The converse of the lemma above, however is not true: take $s$ to be the sequence of identity morphisms.
We next study various special classes of $S$-adic words with which we will work.
\begin{definition}
    A sequence $(\sigma_n)_{n\in\nat}$ over $S(\Sigma)$ is weakly primitive if for every $n$ there exists $m \ge n$ such that $\sigma_n\cdots\sigma_m$ is positive, i.e., for every $b,c \in \Sigma$, $b$ appears in $\sigma_n\cdots\sigma_m(c)$.
\end{definition}
If $(\sigma_n)_{n\in\nat}$ is weakly primitive, then we can compute a sequence $(k_n)_{n\in\nat}$ of increasing integers with $k_0 = 0$ such that $\sigma_{k_n} \cdots \sigma_{k_{n+1}-1}$ is positive for all $k$.
Consequently, for any sequence of letters $(a_n)_{n\in\nat}$ we have that $\lim_{n \to \infty} |\sigma_0 \cdots \sigma_n(a_n)| = \infty$; this is known as being \emph{everywhere growing}.
A word directed by a weakly primitive sequence is uniformly recurrent \cite{durand2003corrigendum}.

The directive sequences with which we will work generate and direct a unique word due to left-properness.
\begin{lemma}
    \label{lem::left-proper}
    Let $(\sigma_n)_{n\in\nat} \in S(\Sigma)^\omega$ be weakly primitive with infinitely many left-proper terms.
    Then there exists unique $\alpha \in \Sigma^\omega$ such that for any $(a_n)_{n\in\nat} \in \Sigma^\omega$, $(\sigma_n,a_n)_{n\in\nat}$ generates $\alpha$.
\end{lemma}
\begin{proof}
	Let $(k_n)_{n\in\nat}$ be an increasing sequence over $\nat$ such that $k_0 = 0$ and $\sigma_{k_n}$ is left-proper for $n > 0$. 
	Further let $b_n \in \Sigma$ be such that all images of $\sigma_{k_n}$ begin with $b_n$. 
	Then $\alpha = \lim_{n\to\infty} \sigma_0 \cdots \sigma_{k_n-1}(a_n)$.
	This limit exists and is infinite because the terms of the sequence are strict prefixes of one another. (In fact, in the parlance of Sec.~\ref{sec::structure-theorems}, the expansion $(\sigma_{k_n}\cdots \sigma_{k_{n+1}-1}, a_{k_n+1})_{n\in\nat}$ is \emph{congenial}.)
	Now take $(\sigma_n, a_n)_{n\in\nat}$, where $(a_n)_{n\in\nat}$ is arbitrary.
	Then $\sigma_0 \cdots \sigma_n(a_n)$ must agree with $\alpha$ on the first $| \sigma_0 \cdots \sigma_{k_{m}-1}(a_m)|$ letters, where $k_m$ is maximal with the property that $k_m \le n$.
	Since $k_m$ becomes arbitrarily large as $n \to \infty$, we have that $\lim_{n\to \infty} \sigma_0 \cdots \sigma_n(a_n) = \alpha$.
\end{proof}

We next illustrate the concepts above through examples.

\subsection{Sturmian and Arnoux-Rauzy words}
\label{subsection:Sturmian}
Let $\eta \in (0, 1) \setminus \mathbb{Q}$. 
The \emph{characteristic} Sturmian word with slope $\eta$ is defined by
$$
\alpha_\eta(n) = \lfloor (n+2) \eta \rfloor - \lfloor (n+1) \eta\rfloor \in \{0,1\}
$$
for $n \in \mathbb{N}$. For instance, the Fibonacci word has $\eta = 1/\phi^2$, where $\phi$ is the golden ratio. Characteristic Sturmian words are $S$-directed for $S = \{\lambda_0, \lambda_1\}$. The substitution $\lambda_i$ maps $i$ to $i$, and the other letter $j$ to $ij$, i.e., it inserts $i$ to the left.

The word $\alpha_\eta$ is intimately connected to the continued fraction expansion of $\eta$.
Suppose $\eta = [0; 1+a_1, a_2, \ldots]$. We then have that $\alpha_\eta$ is the unique word directed by the sequence 
\[
\underbrace{\lambda_0,\ldots,\lambda_0}_{\textrm{$a_1$ times}},\underbrace{\lambda_1,\ldots,\lambda_1}_{\textrm{$a_2$ times}},\underbrace{\lambda_0,\ldots,\lambda_0}_{\textrm{$a_3$ times}},\underbrace{\lambda_1,\ldots,\lambda_1}_{\textrm{$a_4$ times}}, \ldots
\]
which we denote by $s_\eta$.
For example, $1/\phi^2$ has the continued fraction expansion $[0; 2, 1, 1, \ldots]$, and hence the Fibonacci word is directed by $(\lambda_0\lambda_1)^\omega$. 
Observe that $\lambda_0,\lambda_1$ are left-proper.
Moreover, for every $k,m > 0$, $\lambda_0^k\lambda_1^m$ is positive and hence $s_\eta$ is weakly primitive.

A (general) Sturmian word $\alpha$ of slope $\eta \in (0, 1) \setminus \mathbb{Q}$ and intercept $\chi \in [-\eta, 1-\eta]$ is given by one of the following:
\begin{align}
    \alpha(n) &= \lfloor (n+2)\eta + \chi \rfloor - \lfloor (n+1)\eta + \chi \rfloor,  \mbox{ for all } n \\
    \alpha(n) &= \lceil (n+2)\eta + \chi \rceil - \lceil (n+1)\eta + \chi \rceil, \mbox{ for all } n.
\end{align}
Sturmian words are uniformly recurrent, and are equivalently characterised by their factor complexity $p(n) = n+1$, which is the lowest among non-periodic words.
A Sturmian word $\alpha$ with slope $\eta$ and intercept $\chi$ satisfies $\Lcal(\alpha) = \Lcal(\alpha_{\eta})$.
That is, the language of a Sturmian word only depends on its slope.

Sturmian words are not necessarily $S$-adic for $S$ defined above.
However, they are $S$-adic for $S = \{\lambda_0,\lambda_1,\rho_0,\rho_1\}$ where the substitution $\rho_i$ inserts $i$ to the right, i.e, maps $i$ to $i$ and $j$ to $ji$. 
Let $\alpha$ be a Sturmian word with slope $\eta = [0; a_1 + 1, a_2,\ldots]$ and intercept $\chi$.
Then there exists (by \cite[Prop.\ 2.7, also see remark after Thm.\ 2.10]{initial-powers}) an Ostrowski expansion $(b_n)_{n\in\nat}$ of $\chi$ in base $\eta$ such that $\alpha$ is directed by a sequence $(\tau)_{n\in\nat}$ where
\begin{equation}
\tau_n = \lambda_0^{b_{2n+1}}\rho_0^{a_{2n+1} - b_{2n+1}} \lambda_1^{b_{2n+2}}\rho_1^{a_{2n+2} - b_{2n+2}}.
\end{equation} 
Conversely, the rules of Ostrowski expansion guarantee that each $(\tau_n)_{n\in\nat}$ obtained from the expansions $(a_n)_{n\in\nat}$ and $(b_n)_{n\in\nat}$ as above is weakly primitive and directs a Sturmian word.
Observe that we can unpack each $\tau_n$ to obtain a \emph{bona fide} directive sequence over $S$, and every morphism in $S$ is left-proper.

In summary, Sturmian words can characterised as the set of all $\alpha$ generated by some weakly primitive $(\sigma_n)_{n\in\nat} \in \{\lambda_0,\lambda_1,\rho_0,\rho_1\}$ containing infinitely many left-proper terms.
Arnoux-Rauzy words generalise Sturmian words to larger alphabets.
Let $\Sigma = \{0,\ldots,d-1\}$.
For distinct $i,j \in \Sigma$, define $\lambda_i(j) = ij$ and $\rho_i(j) = ji$, and for $i=j \in \Sigma$, let $\lambda_i(j)=\lambda_j(i) = i$.
Observe that each $\lambda_i$ is left-proper.
Then a word $\alpha \in \Sigma^\omega$ is Arnoux-Rauzy if and only if it is generated by a weakly primitive $(\sigma_n)_{n\in\nat} \in \{\lambda_0,\ldots,\lambda_{d-1},\rho_0,\ldots,\rho_{d-1}\}$ containing infinitely many left-proper terms (see, e.g., \cite[Sec.~2.3, Thm.~4.12, Sec.~5]{episturmian-survey}).
Arnoux-Rauzy words have factor complexity $p(n) = (n-1)d+1$, but this is not a characterisation for $d > 2$.
Observe that Sturmian words are precisely the Arnoux-Rauzy words over a two-letter alphabet.
Other $S$-adic generalisations of Sturmian words include episturmian words \cite{episturmian-survey} and dendric shifts \cite{BFFLPR:2015,dendric-ternary}.

\section{\texorpdfstring{Structure theorems for $S$-adic words}{Structure theorems for S-adic words}}
\label{sec::structure-theorems}

Let $s$ be a directive sequence over $S(\Sigma)$ for an alphabet~$\Sigma$.
In this section we will show that any word directed by $s$ can be written as a product of \emph{congenial} words generated by~$s$.
Congenial expansions are one of the main novelties of this paper, and form the combinatorial cornerstone of our analysis of the automaton acceptance problem for $S$-adic words.

\begin{definition}
\label{congenial-definition}
    Let $\Sigma$ be an alphabet.
    A sequence $((\sigma_n,a_n))_{n\in\nat}$ over $S(\Sigma)\times \Sigma$ is congenial if $\sigma_{n+1}(a_{n+1})$ begins with $a_n$ for all $n$. 
    A word $\alpha \in \Sigma^+ \cup \Sigma^\omega$ is $s$-congenial for a directive sequence~$s$ if  $\alpha = \lim_{n \to \infty} \sigma_0  \cdots \sigma_n(a_n)$ for a congenial sequence $((\sigma_n, a_n))_{n\in \mathbb{N}}$ that augments~$s$.
\end{definition}
The most desirable property of a congenial sequence $(\sigma_n,a_n)_{n\in\nat}$ is that $\lim_{n\to \infty} \sigma_0\cdots \sigma_n(a_n)$ is guaranteed to exist, and has every $\sigma_0\cdots \sigma_n(a_n)$ as a prefix.
The following lemma captures this property, and is proven via a straightforward induction.

\begin{lemma}
\label{lem::congenial-limit}
    Let $((\sigma_n, a_n))_{n\in\nat}$ be congenial, and for $n \ge 1$, $v_n \in \Sigma^*$ be such that $\sigma_{n}(a_{n}) =a_{n-1} v_{n}$.
    For all $n$, $\sigma_0\cdots\sigma_{n}(a_n) = \sigma_0(a_0) \cdot \sigma_0(v_1) \cdot \sigma_0\sigma_1(v_2) \cdots (\sigma_0\cdots \sigma_{n-1}(v_n))$
\end{lemma}

For a directive sequence~$s$, denote by $\mathsf{congenials}_s$ the set of all $s$-congenial $\alpha \in \Sigma^\omega$.
We next show that this set is finite.

\begin{lemma}
\label{congenial-existence-finite}
Let $\Sigma$ be an alphabet and $s = (\sigma_n)_{n \in \mathbb{N}}$ be a directive sequence over $S(\Sigma)$.
There exist at least one and at most $|\Sigma|$ congenial expansions of the form $(\sigma_n,a_n)_{n\in\nat}$, and hence $1 \le |\mathsf{congenials}_s| \le |\Sigma|$.
\end{lemma}
\begin{proof}
By Lem.~\ref{generated-implies-directed}, there exists $(\alpha^{(n)})_{n\in\nat}$ over $\Sigma^\omega$ such that $\sigma_0 \cdots \sigma_n(\alpha^{(n+1)}) =  \alpha^{(0)}$ for all $n$.
Let $a_n$ be the first letter of $\alpha^{(n)}$.
We have that $(\sigma_n, a_{n+1})$ is congenial, and hence $|\mathsf{congenials}_s| \ge 1$.

Now suppose there exist $m \ge |\Sigma|+1$ congenial sequences $(\sigma_n, a^{(i)}_n)_{n\in\nat}$.
By a pigeonhole argument, there must exist $i \ne j$ such that $a^{(i)}_n = a^{(j)}_n$ for infinitely many $n$.
From congeniality it follows that $a^{(i)}_n = a^{(j)}_n$ for all $n$. 
\end{proof}

Congenial words constitute the building blocks of directed words. The ``if'' part of the following lemma follows by definition; the ``only if'' part holds because a directed word can naturally be factorised into a congenial prefix and a directed suffix (if the former is finite). 

\begin{lemma}
\label{lem::congenial-concatenate}
Let $\Sigma$ be an alphabet and $s =(\sigma_n)_{n \in \nat}$ be a directive sequence over $S(\Sigma)$.
A word $\alpha \in \Sigma^+ \cup \Sigma^\omega$ is directed by $s$ if and only if it can be expressed as a (possibly infinite) concatenation $u_0u_1\cdots$ of $s$-congenial words.
\end{lemma}
\begin{proof}
	Suppose $\alpha = u_0u_1\cdots$, where $u_i \in \mathsf{congenials}_s$ for all~$i$.
	Let $(\sigma_n, a^{(i)}_n)_{n\in\nat}$ be a congenial sequence generating $u_i$, and $u_i^{(n)} = \lim_{k \to \infty} \sigma_n \cdots \sigma_k(a^{(i)}_k)$ for all $i$.
	We have that $u^{(i)}_n = \sigma_n(u^{(i)}_{n+1})$ for all $n$.
	It remains to define $\alpha^{(n)} = u_0^{(n)}u_1^{(n)}\cdots$.
	Then $\alpha^{(0)} = \alpha$ and $\alpha^{(i)}_n = \sigma_n(\alpha^{(i)}_{n+1})$ for all $n$.
	
	Now suppose $\alpha$ is $s$-directed, and let $(\alpha^{(n)})_{n\in\nat}$ be the witnessing sequence of words with $\alpha^{(0)} = \alpha$.
	Write $a_n$ for the first letter of $\alpha^{(n)}$.
	We construct the desired factorisation inductively. 
	Let $v$ be the word defined by the congenial sequence $(\sigma_n, a_{n+1})_{n \in \nat}$.
	By the choice of $(a_n)_{n\in\nat}$, $v$ is a prefix of $\alpha$.
	If $v = \alpha$, then we are done.
	Otherwise, $v$ must be finite.
	Let $(v_n)_{n\in\nat}$ be the unique sequence of finite words such that $v_0 = v$, $v_n$ is a prefix of $\alpha^{(n)}$ for all $n$, and $\sigma_{n}(v_{n+1}) = v_{n}$ for all $n$.
	Write $\alpha^{(n)} = v_n \gamma^{(n)}$ for all $n$.
	Because $\sigma_n(\alpha^{(n+1)}) = \alpha^{(n)}$ for all $n$ and $\sigma_n(v_{n+1}) = v_n$, we have that $\sigma_n(\gamma^{(n+1)}) = \gamma^{(n)}$.
	That is, $\gamma = \gamma^{(0)}$ is a suffix of $\alpha$ directed by~$s$.
	Set $u_0 = v$, and repeat the process on $\gamma$.
\end{proof}

If $s$ is weakly primitive, then have the following strengthening of Lem~\ref{lem::congenial-concatenate}.
\begin{lemma}
    \label{lem::everywhere-growing-can-congenialise}
    Suppose $s = (\sigma_n)_{n\in\nat}$ is weakly primitive and directs $\alpha \in \Sigma^\omega$.
    Then $\alpha$ has a congenial expansion $(\sigma_n, a_n)_{n\in\nat}$.
\end{lemma}
\begin{proof}
    Let $(\alpha^{(n)})_{n\in\nat}$ be such that $\alpha^{(0)} = \alpha$ and $\sigma_n(\alpha^{(n+1)}) = \alpha^{(n)}$, and $a_n$ be the first letter of $\alpha^{(n+1)}$.
    Then $(\sigma_n, a_n)_{n\in\nat}$ is congenial, $\sigma_0\cdots \sigma_n(a_n)$ is a prefix of $\alpha$ for all $n$, and $\lim_{n\to \infty} \sigma_0\cdots \sigma_n =\infty$ by the growth assumption.
    Therefore, $(\sigma_n,a_n)_{n\in\nat}$ generates~$\alpha$.
\end{proof}

Combining lemmas~\ref{generated-implies-directed} and \ref{lem::everywhere-growing-can-congenialise} we obtain the following.
\begin{lemma}
    \label{lem:eg-sadic=sdirected}
    Let $s$ be a weakly primitive directive sequence.
    A word $\alpha$ is directed by $s$ if and only if it is congenially generated by $s$.
\end{lemma}

\section{Equivalence of substitutions modulo a semigroup}
\label{sec:S-mod-A}
The motivation behind the semigroup-based approach to the language-membership problem is that even though there are infinitely many substitutions in $S(\Sigma)$, from the perspective of a finite $\omega$-semigroup $M_L$ recognising an $\omega$-regular language $L \in \Sigma^\omega$, they can be divided into finitely many equivalence classes.

Let $L \subseteq \Sigma^\omega$ be an $\omega$-regular language, and let it be recognised by a morphism $h_L$ into a finite $\omega$-semigroup $M_L = (M_{L,f}, M_{L,\omega})$, e.g., the morphism into the syntactic $\omega$-semigroup as defined in \cite[Sec.~11]{perrin1995semigroups}. Observe that an $\omega$-semigroup morphism from $\Sigma^\infty$ into $M$ is completely determined by the images of each letter in $M_{L,f}$. Since $M_{L,f}$ is finite, there are only finitely many possible morphisms. We denote the set of these morphisms by $\mathsf{morphisms}_L$. We make a small technical adaptation, and interpret these as monoid morphisms, i.e., we adjoin a fresh neutral element $1_{M_L}$ to $M_L$, and assign $h(\varepsilon) = 1_{M_L}$ for each $h$.

We define an equivalence relation on the set of non-erasing substitutions $\sigma \colon \Sigma^* \to \Sigma^*$.  Let $\mathsf{segments}_\sigma$ be the function that takes a letter $a \in \Sigma$ and $h \in \mathsf{morphisms}_L$, and returns a finite sequence of pairs from $\Sigma \times M_L$, determined as follows.
Write
\[
\sigma(a) = b_1 v_1 \cdots b_d v_d
\]
where $b_1, \dots, b_d$ are distinct letters and $v_i \in \{b_1, \dots, b_i\}^*$ for all $1\le i \le d$. 
I.e., we consider the factorisation of $\sigma(a)$ into segments based on the first occurrence of each letter.
Then
\[
\mathsf{segments}_\sigma(a, h) = \langle (b_1, h(v_1)), \dots, (b_d, h(v_d))\rangle.
\]
Note that there are only finitely many possibilities for $\mathsf{segments}_\sigma$.
For $\sigma,\mu \in S(\Sigma)$, define
\[
\sigma \equiv_L \mu \Leftrightarrow \mathsf{segments}_\sigma = \mathsf{segments}_\mu.
\]
We denote the class of $\sigma$ by $[\sigma]_L$, and the finite set of the equivalence classes by $\Xi_L$.
We next show how to effectively provide representatives for the equivalence classes.

\begin{lemma}
\label{Xi-effective}
    Given an $\omega$-regular language $L$, we can compute morphisms $\sigma_1,\ldots,\sigma_m$ such that $[\sigma_i]_L \ne [\sigma_j]_L$ for all $i \ne j$ and 
    $
    \Xi_L = \{[\sigma_i]_L \colon 1\le i \le m\}.
    $
\end{lemma}
\begin{proof}
	We simply iterate over each syntactic possibility for $\mathsf{segments}_\sigma$, and check if it is realised by a non-erasing substitution. In order to do so, for each letter $a$, we will find a word $w_a$ such that assigning $\sigma(a) = w_a$ is consistent with $\mathsf{segments}_\sigma(a, h)$ for all $h$.
	For $a \in \Sigma$ and $h \in \mathsf{morphisms}_L$ under consideration, let the purported 
	$\mathsf{segments}_\sigma(a, h) = \langle (b_1, x_1),\ldots, (b_k,x_k)\rangle$ with $k \ge 1$ such that $b_i \ne b_j$ for all $i\neq j$ and $x_i \in M_{\Acal}$, $b_i \in \Sigma$ for all $i$.
	We can compute regular languages $L_1,\ldots,L_k \subseteq \Sigma^*$ such that for all $i$ and $w \in\Sigma^*$, $w \in L_i$ if and only if $h(w_i) = x_i$ and $w_i \in \{1,\ldots,b_i\}^*$. Denote $L_{a, h} = b_1 L_1 \cdots b_k L_k$.
	
	We can effectively check whether $L_a = \bigcap_{h \in \mathsf{morphisms}_L} L_{a, h}$ is non-empty, and if yes, effectively compute $w_a \in L_a$. Such a word can be computed as an image for every letter (if and) only if the purported $\mathsf{segments}_\sigma$ is indeed realisable by assigning each $\sigma(a)$ to the corresponding $w_a$.
\end{proof}

For technical convenience, we define the following auxiliary functions that can be derived from $\mathsf{segments}_\sigma$; the first four of them are independent of $L$. 
\begin{itemize}
    \item[(a)] $\mathsf{expanding}_\sigma$ records for each letter $a$ whether $|\sigma(a)| > 1$.
    It evaluates to false if and only if $\mathsf{segments}_\sigma(a, h) = \langle (b, 1_{M_L}) \rangle$ for all $h$.
    \item[(b)] $\mathsf{introduces}_\sigma$ maps each letter to a finite sequence of pairs of letters with Boolean flags: if $\sigma(a) = b_1v_1 \cdots b_dv_d$ when factorised as in the definition of $\mathsf{segments}_\sigma$, then $\mathsf{introduces}_\sigma(a)$ is 
    $\langle (b_1, f_1), \ldots, (b_d, f_d) \rangle$
    where $f_i = 1$ if and only if $v_i \ne \varepsilon$.
    \item[(c)] $\mathsf{head}_\sigma$ maps each letter $a$ to the first letter in $\sigma(a)$.
    \item[(d)] $\mathsf{tail}_\sigma$ takes as input a letter $a$ and a morphism $h$. 
    Write $\sigma(a) = \mathsf{head}_\sigma(a) \cdot v$. 
    We define $\mathsf{tail}_\sigma(a, h) = h(v)$.
    \item[(e)] $\mathsf{compose}_\sigma$ takes $h \in \mathsf{morphisms}_L \to \mathsf{morphisms}_L$ and returns $h \circ \sigma$.
\end{itemize}

We next argue that composition of morphisms can be defined on equivalence classes.
Let $\sigma, \mu \colon \Sigma^*\to\Sigma^*$ be non-erasing.
We show how to determine  $\mathsf{segments}_{\sigma \circ \mu}(a, h)$ using only $\mathsf{segments}_\mu$ and $\mathsf{segments}_\sigma$.
Suppose $\mu(a) = b_1v_1\cdots b_m v_m$, where $b_i \in \Sigma$ and $v_i \in \{b_1,\ldots,b_i\}^*$ for all $i$.
Then
\[
\sigma(\mu(a)) = \sigma(b_1)\sigma(v_1)\cdots \sigma(b_m)\sigma(v_m)
\]
and each letter of $\sigma(v_i)$ will have already appeared in one of $\sigma(b_1),\ldots,\sigma(b_i)$.
Write $\mathsf{segments}_\sigma(b_i,h) = \langle (c_{i,1}, h(w_{i,1})),\ldots,(c_{i, k_i}, h(w_{i, k_i})) \rangle$.
Then $\sigma(\mu(a))$ is equal to 
\[
c_{1,1}w_{1,1}\cdots c_{1,k_1}w_{1,k_1} \sigma(v_1) \cdots c_{m,1}w_{m,1}\cdots c_{m,k_m}w_{m,k_m} \sigma(v_m).
\]
Write $t_i = w_{i,k_i} \, \sigma(v_i)$.
Observe that each letter of $w_{i,j}$ appears as a factor $c_{e,l}$ before $w_{i,j}$ in the factorisation above;
the same applies to every~$t_i$.
To compute $\mathsf{segments}_{\sigma \circ \mu}(a, h)$ we begin with the finite sequence
\begin{multline*}
    (c_{1,1}, h(w_{1,1})),\ldots,(c_{1,k_1-1}, h(w_{1,k_1-1})), (c_{1, k_1}, h(t_1)) \cdots \\
    (c_{m,1}, h(w_{m,1})),\ldots,(c_{m,k_m-1}, h(w_{m,k_m-1})), (c_{m, k_m}, h(t_m)).
\end{multline*}
Note that the above can be effectively computed from $\mathsf{segments}_\sigma$ and $\mathsf{segments}_\mu$ as $h(t_i) = h(w_{i,k_i}) \, h(\sigma(v_i))$
and the two factors can be gleaned from, respectively, $\mathsf{segments}_\sigma(b_i,h)$ and $\mathsf{segments}_\mu(a,\mathsf{compose}_\sigma(h))$.
Rename the indices in the sequence above to obtain $\langle (c_1,h(w_1), \ldots, (c_M, h(w_m)\rangle$ where $c_i \in \Sigma$ and $w_i \in \Sigma^*$ for all $i$.
Recall that $w_i \in \{c_1,\ldots,c_i\}^*$ for all $i$.
But it is possible that $c_i = c_j$ for some $i,j$.
To eliminate these, we repeat the following process for as long as possible.
Find the smallest $j$ such that $c_i = c_j$ for some $i < j$.
Replace the two consecutive terms $(c_{j-1}, h(w_{j-1}), (c_j, h(w_j))$ with $(c_{j-1}, h(w_{j-1}) \, h(c_j) \, h(w_j))$.
In the end we are left with $\mathsf{segments}_{\sigma \circ \mu}(a, h)$.

Observe that $\mathsf{segments}_\sigma$ only depends on $\xi \coloneqq [\sigma]_L$; we can thus index the auxillary functions (a-e) above by the equivalence class $\xi$.
To summarise, we have the following.
\begin{lemma}
    The set $\Xi_L$ is a finite monoid with the binary operation $[\sigma]_L \cdot [\mu]_L = [\sigma \circ \mu]_{L}$ and the identity element $[\mathsf{id}]_L$, where $\mathsf{id}(w) = w$ for all $w \in \Sigma^*$.
\end{lemma}
\begin{proof}
The set of substitutions is a monoid with composition being the binary operation. The map from substitutions to their equivalence classes respects the binary operation, by construction. Using this fact, it is straightforward to check that the binary operation on equivalence classes is associative, and that $[\mathsf{id}]_L$ is indeed the identity element.
\end{proof}

Let $(\sigma_n)_{n\in I}$ be a sequence over $S(\Sigma)$, $\xi_n = [\sigma_n]_L$ for all $n$, and $(a_n)_{n\in I}$ be a sequence of letters from $\Sigma$, where~$I$ can be finite or infinite.
We define $\mathsf{trace}_L((\sigma_n)_{n\in I}) = (\xi_n)_{n\in I}$ and $\mathsf{trace}_L((\sigma_n, a_n)_{n\in I}) = (\xi_n, a_n)_{n\in I}$.
We extend the definition of $[\cdot]_L$ to finite sequences in the natural way: if $I$ is finite, then $[(\sigma_n)_{n\in I}]_L$ is the  product $\prod_{n\in I} \xi_n \in {\Xi_L}$.
We have that for any infinite word~$\alpha$ and substitution~$\sigma$, whether $\sigma(\alpha) \in L$ can be determined from $[\sigma]_L$.

\begin{lemma}
\label{lem:sigma-of-alpha}
Let $L \subseteq \Sigma^\omega$ be an $\omega$-regular language and $\alpha \in \Sigma^\omega$. 
There exists $\Phi \subseteq \Xi_L$ such that for any $\sigma \colon \Sigma^* \to \Sigma^*$, $\sigma(\alpha) \in L$ if and only if $[\sigma]_{\Acal} \in \Phi$. 
Furthermore, $\Phi$ can be effectively computed if we can compute ${h}(\alpha)$ for all $h \in \mathsf{morphisms}_L$.
\end{lemma}
\begin{proof}
Recall \cite[Thm.~7.6]{perrin1995semigroups} that we have a set $G \subseteq M_{L, \omega}$ such that for all $\alpha \in \Sigma^\omega$, $\alpha \in L$ if and only if $h_L(\alpha) \in G$. We now view the substitution $\sigma$ as an $\omega$-semigroup morphism:
\begin{align*}
    h_L(\sigma(\alpha)) = h_L \circ \sigma(\alpha) = \mathsf{compose}_{\sigma}(h_L)(\alpha).
\end{align*}
We thus define $\Phi = \{\xi \in \Xi_L: \mathsf{compose}_\xi(h_L)(\alpha) \in G\}$. The effectiveness claim follows from the fact that $\mathsf{compose}_\xi(h_L) \in \mathsf{morphisms}_L$ for all $\xi$.
\end{proof}

\section{Morphic words in a given language}
\label{section-sigma-omega}

We shall now use the machinery developed in Sec.\ \ref{sec:S-mod-A} to characterise the set of morphic words in a given $\omega$-regular language $L \subseteq \Sigma^\omega$, thus generalising the main result of \cite{carton-thomas}.
A word $\alpha \in \Sigma^\omega$ is \emph{substitutive} if it is a fixed point of a non-trivial substitution, i.e., if there exists $\sigma\in S(\Sigma)$ not the identity such that $\sigma(\alpha) = \alpha$. 
A word $\alpha\in \Sigma^\omega$ is \emph{morphic} if it is of the form $\pi(\alpha)$ for $\pi \in S(\Sigma)$ and $\alpha$ a substitutive word.
For $u \in \Sigma^*$ and substitutions $\sigma, \pi$,  define $\sigma^\omega(u) = \lim_{n\rightarrow \infty} \sigma^n(u)$ and $\pi \circ \sigma^\omega(u) = \pi(\sigma^\omega(u))$.
For such $\alpha$, given an automaton $\mathcal{A}$, we will show that the respective equivalence classes $\xi, \zeta \in \Xi_L$ of substitutions $\sigma, \pi$ determine:
\begin{itemize}
    \item Whether the substitution $\sigma$ has a fixed point in $L$ (Thm. \ref{H-fixed-point}).
    \item A regular language $\Lcal(\xi, \zeta) \subseteq \Sigma^{+}$ such that $\pi \circ \sigma^\omega(u)$ is in $L$ if and only if $u \in \Lcal(\xi, \zeta)$ (Thm.\ \ref{morphic-regular-lang}).
\end{itemize}
Before we proceed to state our key properties, we remark that we work with $\Sigma^\infty \cup \{\bot\}$. For the sake of brevity, we (syntactically) allow infinite words to be concatenated, i.e., for $\alpha \in \Sigma^\omega$ and $\beta \in \Sigma^\infty$, the syntactic concatenation $\alpha \beta$ equals $\alpha$. Furthermore, we syntactically allow for $\bot$ to be concatenated as a word would: for any word $u$, $\bot u = \bot$, the concatenation $u \bot$ is $\bot$ for finite $u$, and is $u$ for infinite $u$.
\begin{lemma}
\label{lemma:sigma-omega}
Let $\Sigma$ be an alphabet and $\sigma$ be a substitution. The following properties hold.
    \begin{description}
    \item[Saturation] For $u \in \Sigma^*$, we have that $\sigma^\omega(u)= \sigma^\omega(\sigma(u))$.
    \item[Distributivity] For $u, v \in \Sigma^*$, $\sigma^\omega(uv) = \sigma^\omega(u) \cdot \sigma^\omega(v)$.  
    \item[Left-expansion] For $a \in \Sigma, u \in \Sigma^+, v \in \Sigma^*$, \\ if $\sigma^\omega(a) = u \cdot \sigma^\omega(a) \cdot \sigma^\omega(v)$, then $\sigma^\omega(a) = u^\omega$.
    \item[Right-expansion] For $a \in \Sigma, u \in \Sigma^*$, if $\sigma(a) = au$, then \\$\sigma^\omega(a) = a \cdot u \cdot \sigma(u) \cdots \sigma^n(u) \cdots$.
    \item[Cycle of Contradiction] Let $a_0, \dots, a_{p-1}$, $p>1$ be distinct letters such that $\mathsf{head}_{\sigma^p}(a_0) = a_0$, and for any $r \in \{1, \dots, p-1\}$, $\mathsf{head}_{\sigma^r}(a_0) = a_r$. We have that\\ $\sigma^\omega(a_0) = \cdots = \sigma^\omega(a_{p-1}) = \bot$.
    \item[Terminal Letters] If, for a letter $a$, $\sigma^\omega(a) = u \in \Sigma^+$, then for all $n \ge |\Sigma|$, $\sigma^n(a) = u$.
\end{description}
\end{lemma}
\begin{proof}
	
	\textbf{Saturation.} The sequences $(\sigma^n(a))_n$ and $(\sigma^{n+1}(a))_{n+1}$ have the same limit.
	
	\textbf{Distributivity.} We have that for all $n$, $\sigma^n(uv) = \sigma^n(u)\sigma^n(v)$. If $\sigma^\omega(u) = \mu \in \Sigma^*$, then for all large $n$, $\sigma^n(uv) = \mu \cdot \sigma^n(v)$. Taking the limit, we get $\sigma^\omega(uv) = \mu \cdot \sigma^\omega(v)$. If $\sigma^\omega(u) \in \Sigma^\omega$, then for every position $j$, there exists $N$ such that for all $n \ge N$, $|\sigma^n(u)| > j$. In other words, every position $j$ is eventually part of $\sigma^n(u)$. Thus, $\sigma^\omega(uv) = \sigma^\omega(u)$. However, if $\beta \in \Sigma^\omega$, then $\beta \mu = \beta$ for all $\mu \in \Sigma^\infty$. If $\sigma^\omega(u) = \bot$, there is some position $j$ such that the letter $\sigma^n(u)(j)$ fluctuates with $n$. This means that the limit of $\sigma^n(u)\sigma^n(v)$ must also be mapped to $\bot$. 
	
	\textbf{Left-expansion.} Follows by repeatedly unrolling the equality.
	
	\textbf{Right-expansion.} Follows by repeatedly applying $\sigma$. 
	
	\textbf{Cycle of Contradiction.} The limit must be $\bot$ as the first letter of $\sigma^n(a_r)$ keeps alternating between $a_0, \dots, a_{p-1}$.
	
	\textbf{Terminal Letters.} We find the set $A$ of letters $a$ such that $\sigma^\omega(a) = u \in \Sigma^+$ by saturation. The key idea is that if $\sigma^\omega(a)$ converges within $n$ iterations, then $\sigma^\omega(\sigma(a))$ must converge within $n-1$ iterations, i.e., for every letter $b$ in $\sigma(a)$, $\sigma^\omega(b)$ must converge within $n-1$ iterations. 
	
	We start with the set $A_0$ of letters $a_0$ such that $\sigma(a_0) = a_0$. We construct $A_{n+1}$ as the union of $A_n$ with the set of letters $a_{n+1}$ such that $\sigma(a_{n+1})$ only contains letters from $A_n$. This construction will saturate within $|\Sigma|$ steps. The invariant is that $A_j$ is the set of $a$ such that $\sigma^\omega(a)$ converges within $j$ iterations. We conclude that since for $n \ge |\Sigma|$, $A_n = A_{|\Sigma|}$, if $\sigma^\omega(a)$ converges in $n$ steps then it must have already converged within $|\Sigma|$ steps.
\end{proof}

The following result, along with distributivity, implies that whether $\sigma^\omega(u)$ for a finite word $u$ is an infinite word accepted by $\mathcal{A}$ is determined by the equivalence class $[\sigma]_L$;
this is the main technical novelty of this section.

\begin{theorem}
    \label{H-converges}
    Let $L \subseteq \Sigma^\omega$ be an $\omega$-regular language. For any $h \in \mathsf{morphisms}_L$, substitutions $\sigma, \tau$ with $\sigma\equiv_L\tau$, and letter $a$, we have that ${h} \circ \sigma^\omega(a) = {h} \circ \tau^\omega(a)$. Moreover, ${h} \circ \sigma^\omega(a)$ can be computed given only the equivalence class $\xi \in \Xi_L$ of $\sigma, \tau$ along with the values of $h, a$.
\end{theorem}
\begin{proof}
	By Lem.\ \ref{lemma:sigma-omega}, we have a dynamic programming algorithm to compute ${h} \circ \sigma^\omega(a)$ in a ``depth-first'' manner. We shall show that this algorithm only uses $[\sigma]_L = \xi$.
	
	Indeed, assume that $\sigma(b) = c_1 v_1 \cdots c_k v_k$ (where the factorisation is based on the first occurrence of each letter; we get $c_1, \ldots, c_k$ from $\mathsf{introduces}_\xi$). By Saturation, we have that
        \[
        {h} \circ \sigma^\omega(b) =  {h} \circ \sigma^\omega(c_1) \cdot 
	{h} \circ \sigma^\omega(v_1) \cdots
	 {h} \circ \sigma^\omega(c_k) \cdot
	 {h} \circ \sigma^\omega(v_k).
        \]
	We observe that the terms $ {h}\circ \sigma^\omega(v_i)$ will affect the result only if all $ {h}\circ \sigma^\omega(c_j)$ for $j \le i$ are elements of $M_{L,f}$. In this case, $\sigma^\omega(v_i)$ will also be a finite word, and by the Terminal Letters property, be $\sigma^{|\Sigma|}(v_i)$. We denote $h \circ \sigma^{|\Sigma|}$ by $g$; thus in this case, we have $ {h}\circ \sigma^\omega(v_i) = g(v_i)$. We can rewrite
	\begin{equation}
		 {h} \circ \sigma^\omega(b) =  {h} \circ \sigma^\omega(c_1) \cdot g(v_1) \cdots  {h} \circ \sigma^\omega(c_k) \cdot g(v_k).
		\label{eq:evaluate-b}
	\end{equation}
    The elements $g(v_1), \ldots, g(v_k)$ can be obtained from $\mathsf{segments}_\xi(b, g)$. 
    If $c_1 = b$, we are in the simpler case of Right-expansion where $\sigma(b) = bu$, and have that 
	\[
            {h}\circ \sigma^\omega(b) = h(b) \cdot h(u) \cdot h\circ\sigma(u) \cdot h\circ\sigma^2(u)\cdots.
    \]
	The first two terms are easily obtained through $\mathsf{tail}_\xi$, and $h \circ \sigma^n(u) = \mathsf{tail}_\xi(b, \mathsf{compose}_{\xi^n}(h))$. Since the sequence of monoid elements $\xi^n$ is effectively ultimately periodic, so is the sequence of factors in the above infinite product, allowing us to compute it.
	
	If $c_1 \ne b$, we first check that there is no $r$ such that $\sigma^r(c_1) = b$: this can be done with access to $\mathsf{head}_{\xi^r}$ for $r \le |\Sigma|$. If this check fails, we have a Cycle of Contradiction, and have that $ {h} \circ \sigma^\omega(b) = \bot$.
	
	We now evaluate expansion \eqref{eq:evaluate-b}. Write $m_0 = 1_{M_L}$, and 
	\[
            m_i =  {h} \circ \sigma^\omega(c_1) \cdot g(v_1) \cdots  {h} \circ \sigma^\omega(c_k) \cdot g(v_i).
    \]
	Clearly, $m_{i+1} = m_i \cdot  {h}\circ \sigma^\omega(c_{i+1})\cdot g(v_{i+1})$, and
        \[
            {h}\circ \sigma^\omega(b) = m_i \cdot  {h} \circ \sigma^\omega(c_{i+1}) \cdot g(v_{i+1}) \cdots  {h} \circ \sigma^\omega(c_k) \cdot g(v_k).
        \]
	In particular when $m_i \notin M_{L,f}$ for some $i$, then $ {h}\circ \sigma^\omega(b) = m_i$. For each $i$, if $c_i \ne b$, we evaluate $ {h}\circ \sigma^\omega(c_i)$ (we make a recursive call if it has not been evaluated before, otherwise we look up the memoized value). Otherwise, $c_i = b$ and we are in the case of Left-expansion, and get that $ {h}\circ \sigma^\omega(b) = m_{i-1}^\omega$. In any case, we will eventually compute $m_k =  {h} \circ \sigma^\omega(b)$.
	
    The result follows by applying the above depth-first routine to compute $ {h} \circ \sigma^\omega(a)$.
\end{proof}

As a corollary, we obtain that whether $\pi(\sigma^\omega(u))$ belongs to an $\omega$-regular language $L$ is completely determined by the equivalence classes $[\sigma]_L$ and $[\pi]_L$: indeed, acceptance only depends on ${h}_L\circ \pi\circ \sigma^\omega(u)$, which is the same as $h = {\mathsf{compose}_{[\pi]_L} (h_L)}$ applied to $\sigma^\omega(u)$.

\begin{corollary}
\label{morphic-regular-lang-helper}
Let $L \subseteq \Sigma^\omega$ be an $\omega$-regular language. For any letter $a$ and $h \in \mathsf{morphisms}_L$ we can compute ${h} \circ \pi \circ \sigma^\omega(a)$ given only $h, a, [\sigma]_L, [\pi]_L$.
\end{corollary}

We thus arrive at the following.
\begin{theorem}
\label{morphic-regular-lang}
Let $L \subseteq \Sigma^\omega$ be an $\omega$-regular language and $\sigma, \pi$ be substitutions with respective equivalence classes $\xi, \zeta \in \Xi_L$. We can compute a regular language $\Lcal(\xi, \zeta) \subseteq \Sigma^+$ such that $\pi \circ \sigma^\omega(u) \in L$ if and only if $u \in \Lcal(\xi, \zeta)$.
\end{theorem}
\begin{proof}
We construct a deterministic finite-word automaton recognising $\Lcal(\xi, \zeta)$. The set of states are the elements of $M_{L,f} \cup M_{L,\omega}$, the initial state is $1_{M_L}$, the set of accepting states is~$G$ (the semigroup elements whose preimages comprise $L$), and the transition function maps $(m, a)$ to $m \cdot ({h_L}\circ \pi \circ \sigma^\omega)(a)$. 
The latter is effective by Cor.\ \ref{morphic-regular-lang-helper}.
\end{proof}

Finally, whether $\sigma$ has a fixed point in $L$ is effectively determined by $[\sigma]_L$.

\begin{theorem}
    \label{H-fixed-point}
    Let $L \subseteq \Sigma^\omega$ be an $\omega$-regular language and $\sigma$ be a substitution. Given $[\sigma]_L$ and $h \in \mathsf{morphisms}_L$, we can compute the set $\{{h}(\alpha) \colon \sigma(\alpha) = \alpha\} \subseteq {M_L}$.
\end{theorem}
\begin{proof}
We observe that any fixed point of $\sigma$ must be a concatenation of words $\sigma^\omega(a)$, where $\sigma(a) = au$ for $u \in \Sigma^*$.
Write $A$ for the set of all such letters $a$, which can be extracted from $\mathsf{introduces}_\sigma$.
The required set is then the sub-$\omega$-semigroup of ${M_L}$ generated by $\{{h} \circ \sigma^\omega(a) \colon a \in A\}$, whose elements in turn can be computed using Thm.\ \ref{H-converges}.
\end{proof}

\section{\texorpdfstring{$S$-adic words in a given language}{S-adic words accepted by A}}
\label{sec::aut-acceptance}

We now present our main results, i.e.,  solutions to the language membership problem for $S$-adic words. 
In this section, let $L \subseteq \Sigma^\omega$ be an $\omega$-regular language recognised by a morphism $h_L$ into a finite $\omega$-semigroup $M_L$, and let $H \subseteq M_{L, \omega}$ be the elements whose preimages constitute $L$. Recall $\Xi_L$ the set of equivalence classes of substitutions modulo $M_L$, as defined in Sec.~\ref{sec:S-mod-A}.
We will first prove that for congenial $\hat{s}$ over $S(\Sigma) \times \Sigma$, defining a word accepted by $\Acal$ is a property of $\mathsf{trace}_L(\hat{s})$, and the set of all such traces is $\omega$-regular.

\begin{theorem}[Main Result for Generated Words]
    \label{cor::congenial-automaton-version} 
Let $L \subseteq \Sigma^\omega$ be an $\omega$-regular language. 
   We can compute an automaton $\Bcal$ over $\Xi_L \times \Sigma$ such that for all infinite sequences $\hat s$ over $S(\Sigma) \times \Sigma$, $\Bcal$~accepts $\mathsf{trace}_L(\hat s)$ if and only if $\hat s$ is congenial and generates a word in $L$.
\end{theorem}
\begin{proof}
   The automaton $\Bcal$ accepts $\hat  s \in (\Xi_L \times \Sigma)^\omega$ if and only if $\hat s$ is accepted by $\Bcal_x$ (constructed in Lem.~\ref{lem::congenial-monoid-version} below) for some $x \in H$. Intuitively, $\Bcal_x$, upon reading $\hat s$, checks that the expansion is congenial, and uses the property to map consecutive factors $u_n$ of the generated word to $h_L(u_n) \in M_{L,f} \cup \{1_{M_L}\}$. It then simulates the run of $\Acal_x$ (from Lem.~\ref{lem::infinite-product-in-monoids-via-automata}) on this stream of images.
\end{proof}

\begin{lemma}
    \label{lem::congenial-monoid-version}
    Let $x \in M_{L,f} \cup M_{L, \omega}$.
    We can construct an automaton $\Bcal_x$ over $\Xi_L \times \Sigma$ such that for all infinite sequences $\hat s$ over $S(\Sigma) \times \Sigma$, $\Bcal_x$ accepts $\mathsf{trace}_L(\hat s)$ if and only if $\hat s$ is congenial and ${h_L}(\alpha) = x$, where $\alpha$ is the word generated by $\hat s$.
\end{lemma}
\begin{proof}
	Recall that a sequence $s = (\sigma_n,a_n)_{n\in\nat}$ over $  \Xi_L \times \Sigma$ is congenial if and only if $\mathsf{head}_{\xi_n}(a_n) = a_{n-1}$ for all $n \ge 1$, where $\xi_n = [\sigma_n]_L$.
	This property depends only on $\mathsf{trace}_L(s)$.
	The automaton $\Bcal_x$, first and foremost, checks the condition above for all $n\ge 1$, and permanently transitions into a rejecting state if it observes violating $\xi_n, a_n, a_{n-1}$.
	
	Now suppose $s = (\sigma_n,a_n)_{n\in\nat}$ is congenial, and define $\xi_n$ as above.
	As shown in Lem.~\ref{lem::congenial-limit}, 
	\begin{equation*}
		\sigma_0\cdots\sigma_{n}(a_n) = \sigma_0(a_0) \cdot \sigma_0(v_1) \cdot \sigma_0\sigma_1(v_2) \cdots (\sigma_0\cdots \sigma_{n-1}(v_n))
	\end{equation*}
	for all $n$, where $v_n$ satisfies $\sigma_n(a_n) = a_{n-1}v_n$.
	Let $u_0 = \sigma_0(a_0)$ and $u_n = \sigma_0\cdots \sigma_{n-1}(v_n)$ for $n \ge 1$.
	By the properties of infinite products in $M_L$ (see Sec.~\ref{sec::aut-semigroups}) we have that ${h_L}(\alpha) = x$ if and only if $\prod_{n=0}^\infty h_L(u_n) = x$.
	The automaton $\Bcal_x$ simulates the run of the automaton $\Acal_x$ of Lem.~\ref{lem::infinite-product-in-monoids-via-automata} on the sequence $(h_L(u_n))_{n\in\nat}$.
	It remains to show how the automaton keeps track of $h_L(u_n)$ as it reads $(\xi_n, a_n)_{n\in\nat}$.
	We have that
	\begin{align*}
		h_L(u_0) 
		&= (\mathsf{compose}_{\xi_0}(h_L)) (a_0), \\
		h_L(u_{n}) &=  \mathsf{tail}_{\xi_{n}}(a_n, \mathsf{compose}_{\xi_0\cdots\xi_{n-1}}(h_L))
	\end{align*}
	for $n \ge 1$.
	The automaton $\Bcal_x$ keeps track of one piece of information $\xi \in \Xi_L$, in addition to the state required for simulating a run of $\Acal_x$ on $(h_L(u_n))_{n\in\nat}$.
	Before reading $(\xi_{n}, a_n)$ the value of $\xi$ is $\xi_0\cdots \xi_{n-1}$, where the empty product (corresponding to the initial value of $\xi$) is the identity element of $\Xi_L$.
	Upon reading $(\xi_n,a_n)$, the automaton $\Bcal_x$ first computes $h_L(u_n)$, in which the value of $\xi$ is used, then feeds the computed value to $\Acal_x$, and finally updates $\xi$ to $\xi \cdot \xi_n$.
	Finally, $\Bcal_x$ accepts $(\sigma_n,a_n)_{n\in\nat}$ if and only if $\Acal_x$ accepts $(h_L(u_n))_{n\in\nat}$.
\end{proof}

When $S$ is finite, for every individual $\sigma \in S(\Sigma)$ we can compute the equivalence class~$[\sigma]_L$, which yields the following. 
\begin{corollary}
    \label{cor::congenial-main-finite}
    Let $S \subseteq S(\Sigma)$ be finite.
    We can compute an automaton $\Bcal$ over $S \times \Sigma$ such that $\Bcal$ accepts $\hat s$ if and only if $\hat s$ is congenial and generates a word in $L$.
\end{corollary}

We next consider $S$-directed words, which are products of congenial words by Lem.~\ref{lem::congenial-concatenate}.

\begin{theorem}[Main Result for Directed Words]
\label{desubstitutible-dream}
Let $L \subseteq \Sigma^\omega$ be an $\omega$-regular language.
We can construct an automaton $\mathcal{B}$ over~$\Xi_L$ such that for all $s \in S(\Sigma)^\omega$, $\Bcal$ accepts $\mathsf{trace}_L(s)$ if and only if $s$ directs some $\alpha \in L$.
\end{theorem}
\begin{proof}
Denote the image of $\mathsf{congenials}_s$ under $h_L$ by $X_s \subset M_{L, f} \cup M_{L, \omega}$. 
Recall from Lem.~\ref{lem::congenial-concatenate} that directed words are obtained by concatenating congenial words. To prove Thm.~\ref{desubstitutible-dream}, we need to recognise the set of all $\mathsf{trace}_s$ for which the sub-$\omega$-semigroup generated by $X_s$ intersects the accepting set~$H$. 
We can precompute a set $\Xcal$ of sets $X$ that generate sub-$\omega$-semigroups intersecting~$H$. 
Our automaton $\Bcal$ needs to check that $X_s$ contains at least one such set $X$. We denote by $\Ccal_x$ the projection of $\Bcal_x$ from Thm.~\ref{lem::congenial-monoid-version} to $\Xi_L$, and observe:
$
s \in \Lcal(\Bcal)$ if and only if $\bigvee_{X \in \Xcal}\bigwedge_{x \in X} x \in \Lcal(\Ccal_x)$.
\end{proof}
\begin{corollary}
    \label{cor::generated-main-finite}
    Let $S \subseteq S(\Sigma)$ be finite.
    We can compute an automaton $\Bcal$ over $S$ such that $\Bcal$ accepts $s$ if and only if $s$ directs a word in $L$. 
\end{corollary}
The corollary above is proven in the same way as Cor.~\ref{cor::congenial-main-finite}.
We can apply the former to Arnoux-Rauzy words (which subsume Sturmian words) over the alphabet $\Sigma = \{0,\ldots,d-1\}$. 
Let $S = \{\lambda_0,\ldots,\lambda_{d-1},\rho_0,\ldots,\rho_{d-1}\}$.
Recall that a word $\alpha$ is Arnoux-Rauzy if and only if it is directed by a (i) weakly primitive $s \in S^\omega$ which (ii) contains infinitely many left-proper terms; by Lem.~\ref{lem::left-proper}, such directive sequences direct a unique word, which is infinite.
By inspection, $s \in S^\omega$ is weakly primitive if and only if for every $0 \le i < d$, either $\rho_i$ or $\lambda_i$ occurs infinitely often.
Hence the requirements (i-ii) can be checked by an automaton, and from Cor.~\ref{cor::generated-main-finite} we obtain the following.

\begin{theorem}[Main Result for Arnoux-Rauzy words]
    \label{thm::main-wp}
    Let $\Sigma = \{0,\ldots,d-1\}$, $L \subseteq \Sigma^\omega$ be an $\omega$-regular language, and $S$ be the finite set of morphisms generating the Arnoux-Rauzy words over $\Sigma$ as described above.
    We can compute an automaton $\Ccal$ that accepts $s \in S^\omega$ if and only if $s$ is weakly primitive, contains infinitely many left-proper terms, and the unique Arnoux-Rauzy word directed by $s$ is in $L$.
\end{theorem}

\section{Partial Quotients}
\label{sec:partial-quotients}
In this section, we refine our main results presented in Sec.~\ref{sec::aut-acceptance}.
For characteristic Sturmian words, in particular, we will show that whether a deterministic parity automaton $\Acal$ accepts the characteristic Sturmian word  $\alpha_\eta$ with slope $\eta$  only depends on the first $N(\Acal)$ terms in the continued fraction expansion of $\eta$, where $N$  is independent of the slope  $\eta$.

Fix a (deterministic) automaton $\Acal$, a class of words $\Wcal \subseteq \Sigma^\omega$, and $S \subseteq S(\Sigma)$ such that
\begin{enumerate}
    \item $\alpha \in \Wcal$ if and only if $\alpha$ is directed by some weakly primitive $s$ over $S^\omega$, and
    \item there exists an effectively computable function $p$ such that $p_\alpha(n) = p(n)$ for all $\alpha \in \Wcal$ and $n \in\nat$, where $p_\alpha(n)$ is the number of distinct factors of $\alpha$ of length $n$.
\end{enumerate}
By 1, every $\alpha \in \Wcal$ is uniformly recurrent.
Characteristic Sturmian, Sturmian, and $d$-letter Arnoux-Rauzy words are captured, respectively, with $S = \{\lambda_0,\rho_0\}$, $p(n) = n+1$; $S = \{\lambda_0,\rho_0,\lambda_1,\rho_1\}$, $p(n) = n+1$; and $S = \{\lambda_0,\ldots,\lambda_{d-1},\rho_0,\ldots,\rho_{d-1}\}$, $p(n) = n(d-1)+1$.

The main idea is that every $\alpha \in \mathcal{W}$ has a prefix $u$ that is \emph{$p$-saturated} with respect to~$\mathcal{A}$, i.e., any uniformly recurrent extension $\beta$ of $u$ with factor complexity $p$ agrees with $\alpha$ upon acceptance by~$\mathcal{A}$ (Lem.~\ref{lem::saturated-prefix}). 
The proof involves a careful consideration of Sem\"enov's algorithm for deciding whether a given automaton accepts a given uniformly recurrent word.
\begin{lemma}
\label{lem::saturated-prefix}
Let $\alpha \in \Wcal$ with effectively computable letters.
We can compute $N$ such that any $\beta \in \Wcal$ with $\beta[0,N)=\alpha[0,N)$ is accepted by $\Acal$ if and only if $\alpha$ is accepted by $\Acal$.
\end{lemma}
\begin{proof}
	Since we have access to the factor complexity function of $\alpha$ (and $\beta$), we can evaluate $R_\alpha(n)$ for all $n$: given $n$, enumerate prefixes $\alpha(0, L)$ until finding $M$ such that (i) $\gamma[0, L)$ contains $p_\alpha(M)$ distinct factors of length $M$, and (ii) all of these $p_\alpha(M)$ factors contain $p_\alpha(n)$ distinct factors of length $n$. 
	Then $M = R_\alpha(n)$.
	
	Applying Prop.~\ref{prop::semenov-and-prefix}, we obtain $N_1$ and an $\ell$ such that if $\alpha[0, N_1) = \beta[0, N_1)$ and $R_\alpha(n) = R_\beta(n)$ for $n \le \ell$, then $\alpha$ and $\beta$ agree upon acceptance by~$\mathcal{A}$. By the above paragraph, we can compute $N_2$ such that if $\alpha[0, N_2) = \beta[0, N_2)$, then $R_\alpha(n) = R_\beta(n)$ for $n \le \ell$. 
	It remains to take $N = \max(N_1, N_2)$.
\end{proof}

Given a weakly primitive and congenial expansion $(\sigma_n, a_n)_{n\in\nat}$ of $\alpha$, we can compute an increasing sequence $(k_n)_{n\in\nat}$ with $k_0 = 0$ such that, writing $l_n = k_{n+1}-1$, $\tau_n = \sigma_{k_n}\cdots\sigma_{l_n}$ and $b_n = a_{l_n}$, each $\tau_n$ is positive, all strict prefixes of the composition $\sigma_{k_n}\cdots\sigma_{l_n}$ are not positive, and $(\tau_n, b_n)_{n\in\nat}$ is also a congenial expansion of $\alpha$.
We refer to $(\tau_n, b_n)_{n\in\nat}$ as the sequence of \emph{partial quotients} of $(\sigma_n, a_n)_{n\in\nat}$. By construction, the sequence of partial quotients is weakly primitive.
The main result of this section is that, for the class $\Wcal$, acceptance by $\Acal$ is determined by a few initial partial quotients; for reasons of space, we delegate the proof, which combines algebraic, combinatorial, and topological reasoning, to the appendix.

\begin{theorem}
    \label{thm::partial-quotients}
    We can compute $N$ with the following property.
    Let $\alpha, \alpha'$ be in $\Wcal$ with congenial $S$-adic expansions $(\tau_n,a_n)_{n\in\nat}, (\tau'_n,a'_n)_{n\in\nat}$ and partial quotients $(\pi_n, b_n)_{n \in \mathbb{N}}$ and $(\pi'_n, b'_n)_{n \in \mathbb{N}}$, respectively.
    If $\pi_n \equiv_L \pi'_n$ and $b_n = b'_n$ for all $n \le N$, then $\Acal$ accepts $\alpha$ if and only if it accepts $\alpha'$.
\end{theorem}
\begin{proof}
	Let $P$ be the set of all positive $\sigma_1\cdots\sigma_r$ such that $\sigma_i\in S$ for all $i$ and $\sigma_1\cdots \sigma_i$ is not positive for all $i < r$.
	Construct finite $\Pi \subseteq P$ such that $\{[\pi]_L \colon \pi \in P\} = \{[\pi]_L \colon \pi  \in \Pi\}$, and let $\Omega$ be the set of all congenial $(\pi_n, b_n)_{n\in\nat}$ over $\Pi~\times~\Sigma$.
	Because congeniality is a local condition,  $\Omega$ is a compact subset of $(\Pi \times \Sigma)^\omega$.
	Next, consider $\alpha$ generated by some $\hat s = (\pi_n, b_n)_{n\in\nat} \in \Omega$.
	By Lem.~\ref{lem::saturated-prefix} there exists $M$ such that $\alpha[0,M)$ is $p$-saturated.
	Write $\Omega_M$ for the set of all $\hat s \in \Omega$ whose first $M$ terms generate a $p$-saturated finite word, observing that each $\Omega_M$ is open.
	From Lem.~\ref{lem::saturated-prefix} it follows that $\{\Omega_M \colon M \in \nat \}$
	is an open cover of $\Omega$, which, by compactness, admits a finite sub-cover.
	That is, there exists $N$ (which we can be effectively computed by enumeration) such that for every $(\pi_n,b_n)_{n\in\nat} \in \Omega$, $(\pi_n,b_n)_{n=0}^N$ generates a finite word that is $p$-saturated.
	
	Now consider $\alpha \in \Wcal$ with a weakly primitive and congenial $S$-adic expansion $(\mu_n, a_n)_{n\in\nat}$ and partial quotient sequence $\hat s \coloneqq (\tau_n, b_n)_{n\in\nat}$.
	Let $(\pi_n)_{n\in\nat}$ over $\Pi$ be such that $\tau_n \equiv_L \pi_n$ for all $n$, and observe that $\hat t \coloneqq (\pi_n,b_n)_{n\in\nat}$ is also congenial.
	Let $\beta$ be the word generated by $\hat{t}$.
	By Thm.~\ref{lem::congenial-monoid-version}, $\Acal$~accepts $\alpha$ if and only if it accepts $\beta$.
	In particular, acceptance by $\Acal$ only depends on the trace $([\tau_n]_{\Acal}, b_n)_{n\in\nat}$.
	By the earlier argument, whether $\Acal$ accepts $\beta$ only depends on 
	$([\pi_n]_L,b_n)_{n\in\nat}$.
	Combining the argument of $p$-saturation, it follows that whether $\Acal$ accepts $\alpha$ only depends on $([\tau_n]_L,b_n)_{n=0}^N$.
\end{proof}
Let us illustrate this result on Sturmian words.
Given an automaton $\Acal$, apply Thm.~\ref{thm::partial-quotients} with $S = \{\lambda_0,\lambda_1,\rho_0,\rho_1\}$ and $p(n)=n+1$ to compute $N$.
Let $\eta \in (0,1)\setminus \rat$, $\chi \in [-\eta, 1-\eta]$, $[0;a_1+1,a_2,\ldots]$ be the continued fraction expansion of $\eta$, $(b_n)_{n\in\nat}$ be an Ostrowski expansion of $\chi$ in base $\eta$, and $\alpha$ be the corresponding Sturmian word with slope $\eta$ and intercept $\chi$.
Recall that $\alpha$ is directed (and uniquely generated; see Lem.~\ref{lem::left-proper}) by the sequence 
\[
(\sigma_n)_{n\in\nat} = 
\langle \lambda_0^{b_1}, \rho_0^{c_1}, \lambda_1^{b_2}, \rho_1^{c_2}, \lambda_0^{b_3}, \rho_0^{c_3}, \lambda_1^{b_4}, \rho_1^{c_4} \ldots \rangle
\]
where $c_i = a_i - b_i$.
By the rules of Ostrowski expansion, at least one of $c_n,c_{n+1}$ is non-zero for all $n$.
Moreover, every composition of morphisms from $\{\lambda_0,\lambda_1,\rho_0,\rho_1\}$ that includes two morphisms with differing indices is positive.
Since at least one of $b_n,c_n$ is non-zero for all $n > 1$ (we could, however, have $a_1 = b_1 = c_1 = 0$), we have that for every~$n$, $\sigma_n \cdots \sigma_{n+5}$ is positive.
Applying Thm.~\ref{thm::partial-quotients}, whether $\alpha$ is accepted by $\Acal$ can be determined by looking at the first $6N$ digits of the expansions of $\eta$ and~$\chi$, since these are guaranteed to generate at least the first $N$ partial quotients of the unique $S$-adic expansion of $\alpha$.


\bibliography{refs}

\end{document}